\documentclass[runningheads]{llncs}

\usepackage[utf8]{inputenc}
\usepackage{amsmath,amssymb}
\usepackage{mathtools}
\usepackage{microtype}
\usepackage{enumitem}
\usepackage{float}
\usepackage{hyperref}
\usepackage[capitalize,noabbrev]{cleveref}

\newcommand{\E}{\mathbb{E}}
\newcommand{\vH}{v_H}
\newcommand{\vL}{v_L}
\newcommand{\dtil}{\tilde d}
\newcommand{\ask}{\mathrm{ask}}
\newcommand{\bid}{\mathrm{bid}}

\title{The Privacy Subsidy in Glosten-Milgrom: \\
Bid-Ask Spread and Welfare under Flip-Noise Direction Observation}
\titlerunning{The Privacy Subsidy in Glosten-Milgrom}
\author{Yuki Nakamura\orcidID{0009-0001-7174-6737}}
\authorrunning{Y. Nakamura}
\institute{The Open University of Japan}

\begin{document}
\maketitle

\begin{abstract}
\begin{sloppypar}
We derive a closed-form bid-ask spread and welfare decomposition
for the Glosten-Milgrom 1985 sequential-trading model when the
market maker observes the trade direction perturbed by a binary
flip channel of probability $\eta$ --- a natural information-theoretic
model of privacy mechanisms acting on the direction signal.
Under a committed Bayesian market-maker pricing rule, the
equilibrium spread is $\mu(1-2\eta) \Delta$, where $\mu$ is the
informed-trader fraction and $\Delta = v_H - v_L$ the value range.
The welfare decomposition identifies a per-trade transfer
$\mu\eta\Delta$ from the protocol's liquidity pool to traders ---
the \emph{privacy subsidy}, mirroring the Gaussian-Kyle analog
established in prior work. The result extends the privacy-subsidy
concept from continuous Gaussian to discrete two-state
microstructure, demonstrating robustness across both classical
models. Primary application: MPC-based matching engines with
$\varepsilon$-differentially-private direction disclosure, where
the engine prices on a noisy direction signal.
\end{sloppypar}

\end{abstract}

\section{Introduction}
\label{sec:intro}

The Glosten-Milgrom (GM) 1985 model is the canonical
discrete-state account of how bid-ask spreads arise from
information asymmetry between informed and uninformed traders
\cite{glosten1985}. A risk-neutral market maker (MM) sets bid and
ask prices conditional on the observed direction of an arriving
trade, and the resulting spread is the adverse-selection cost the
MM must charge to break even against the informed trader.

\begin{sloppypar}
In the textbook GM model, the MM observes the trade direction
$d \in \{\text{buy}, \text{sell}\}$ exactly. In contemporary
privacy-preserving exchange designs --- sealed-bid request-for-quote
(RFQ) systems, multi-party computation (MPC)-based matching,
zero-knowledge order routing --- the MM's signal of trade direction
may be noised by the privacy mechanism. We model this leakage
information-theoretically as a binary symmetric channel: the MM
observes $\dtil = d$ with probability $1 - \eta$ and the opposite
direction with probability $\eta \in [0, 1/2]$.
\end{sloppypar}

Under the natural \emph{committed Bayesian} MM pricing rule (the
MM prices via $\ask = \E[v \mid \dtil = \text{buy}]$ and
$\bid = \E[v \mid \dtil = \text{sell}]$, not constrained to
zero-profit on actual trade direction), we derive in closed form:

\begin{enumerate}[leftmargin=*]
\item The modified spread:
$\mathrm{spread} = \mu(1 - 2\eta) \cdot \Delta$
where $\mu$ is the informed-trader fraction and $\Delta = \vH - \vL$
the asset's value range (\Cref{thm:spread}).

\item The welfare decomposition:
$\mu\pi_I + (1-\mu)\pi_N + \pi_M = 0$ with closed-form per-trade quantities,
yielding a \emph{privacy subsidy}
$|\pi_M^{\mathrm{GM}}| = \mu \eta \Delta$ paid by the
protocol/LP pool to traders (\Cref{thm:welfare}).

\item Comparative statics and a counter-intuitive corollary
mirroring our prior Kyle / Gaussian analog
\cite{nakamura2026privacysubsidy}: noise traders gain from
privacy alongside informed traders, with the protocol bearing
the entire cost.
\end{enumerate}

\paragraph{Positioning.}
Our result extends the privacy-subsidy concept established for
continuous-Gaussian Kyle markets~\cite{nakamura2026privacysubsidy}
to discrete two-state GM. The closest contemporary work is the
information-thermodynamic analysis of GM by Touzo, Marsili, and
Zagier~\cite{touzo2020information}, which derives a market-
temperature bound on informed-trader gain under \emph{exact} MM
observation. Our flip-noise extension yields a concrete closed-form
welfare quantity (the privacy subsidy) under noisy observation;
the relationship to their thermodynamic upper bound under flip
noise is suggestive but not pursued here.

\paragraph{Roadmap.}
\Cref{sec:related} positions against the relevant literature.
\Cref{sec:setup} sets up the model. \Cref{sec:spread} states and
proves the modified bid-ask spread (\Cref{thm:spread}).
\Cref{sec:welfare} derives the welfare decomposition and
identifies the privacy subsidy (\Cref{thm:welfare}).
\Cref{sec:cs} gives comparative statics including the counter-
intuitive noise-trader corollary. \Cref{sec:application} maps the
result to MPC- and RFQ-style privacy mechanisms.
\Cref{sec:lvr} places the privacy subsidy in the same family as
Loss-Versus-Rebalancing. \Cref{sec:conclusion} concludes.

\section{Related work}
\label{sec:related}

\paragraph{GM lineage.}
Glosten and Milgrom~\cite{glosten1985} establish the two-state
sequential-trading model with a competitive zero-profit MM.
Glosten and Harris~\cite{glosten1988} provide an empirical
decomposition of the spread.
Easley and O'Hara~\cite{easleyohara1992} analyze information
revelation through trade timing.
Brahma et al.~\cite{brahma-bayesian-mm} introduce a sequential Bayesian
MM that updates beliefs after each trade; our committed
Bayesian-AMM pricing rule inherits this lineage but departs from
zero-profit (motivated by smart-contract / committed-mechanism
interpretation).
Das~\cite{das2005learning} considers an online-learning MM in the
GM model; algorithmic, distinct from our closed-form result.
Huddart, Hughes, and Levine~\cite{huddart2001disclosure} modify
Kyle with ex-post insider disclosure --- the dual of our setup
(trader-side imperfection vs MM-side).

\paragraph{Information-theoretic GM.}
Touzo, Marsili, and Zagier~\cite{touzo2020information} map GM to
a Szilárd information engine: they define a market temperature
from the distribution of orders and bound the informed trader's
expected gain by temperature times information. Their analysis
uses exact MM observation and yields a thermodynamic upper bound;
ours derives a concrete closed-form welfare loss under a specific
binary-symmetric-channel privacy mechanism. The two results are
complementary in scope --- their bound applies in any GM setting
with exact observation, while ours quantifies the protocol cost
under a specific noise channel. Carmier~\cite{carmier2022thermo}
extends the Touzo--Marsili--Zagier framework to finite horizon by
deriving an upper bound on informed-trader expected gain that is
fully analogous to a generalised second law of thermodynamics. A
precise comparison (whether our $\mu\eta\Delta$ subsidy saturates
Carmier's finite-horizon bound under flip noise, with channel
capacity reduced from one bit to $1 - h(\eta)$ bits) is a natural
next step and left to future work.

\paragraph{Anonymity and dark-pool literature.}
Zhu~\cite{zhu2014darkpools} analyzes dark-pool trading as a
routing equilibrium between lit and dark venues; Buti, Rindi, and
Werner~\cite{buti2017darkpool} extend this. The mechanism is
routing-based segmentation, mathematically distinct from our
within-venue direction noising.

\paragraph{Connection to the Kyle / Gaussian-noise companion.}
Nakamura~\cite{nakamura2026privacysubsidy} establishes the
privacy-subsidy concept for continuous-Gaussian Kyle markets
under additive observation noise. The result is
$|\pi_M^{\mathrm{Kyle}}| = \sigma_v \sigma_\varepsilon^2 /
(2\sqrt{\sigma_u^2 + \sigma_\varepsilon^2})$. The present paper
provides the discrete-GM analog under binary-channel noise,
$|\pi_M^{\mathrm{GM}}| = \mu \eta \Delta$.

\section{Model setup}
\label{sec:setup}

A single risky asset has value $v \in \{\vH, \vL\}$ with prior
$P(v = \vH) = \pi$. We work in the symmetric case $\pi = 1/2$
throughout; the asymmetric extension is routine and reported in
\Cref{rem:asymmetric}. Define $\Delta := \vH - \vL > 0$.

\paragraph{Trader arrival.}
At each round, a single trader arrives. With probability
$\mu \in (0, 1)$, the trader is \emph{informed} (knows $v$); with
probability $1 - \mu$, the trader is \emph{uninformed} (a noise
trader). The informed trader chooses direction
$d \in \{\text{buy}, \text{sell}\}$ optimally; the noise trader
chooses direction uniformly at random.

\paragraph{Privacy channel.}
The MM does not observe the actual direction $d$. Instead, the MM
observes $\dtil \in \{\text{buy}, \text{sell}\}$, which equals $d$
with probability $1 - \eta$ and the opposite direction
with probability $\eta \in [0, 1/2]$. The privacy parameter $\eta$
is the flip-noise probability of the binary symmetric channel
modeling the privacy mechanism.

\paragraph{Timing and pricing rule.}
The MM commits ex ante to a two-sided quote whose levels are set
from the Bayesian posterior given the noisy direction signal:
\begin{align}
\ask(\dtil = \text{buy})
  &= \E[v \mid \dtil = \text{buy}],
\quad
\bid(\dtil = \text{sell})
  = \E[v \mid \dtil = \text{sell}].
\end{align}
An arriving trader transacts on its \emph{true} direction $d$: a
buy executes at the ask, a sell at the bid. The privacy channel
does not change which side the trader fills; it only degrades how
informative $\dtil$ is about $d$, moving both quote levels toward
the prior midpoint (the spread collapses to zero at $\eta = 1/2$).
The MM is a \emph{committed Bayesian AMM}: it commits to this rule
and bears the realized expected loss on the executed flow, without
a zero-profit constraint. Relative to classical GM, in which the
posted quote levels are set from $\E[v \mid d]$, here they are set
from the noise-coarsened posterior $\E[v \mid \dtil]$; the
adverse-selection mechanism is the same, only the maker's
information is coarsened. The committed-Bayesian framing parallels
the Kyle / Gaussian-noise companion~\cite{nakamura2026privacysubsidy}
(see \Cref{rem:competitive}).

\section{The modified spread}
\label{sec:spread}

\begin{theorem}[Spread under flip-noise observation]
\label{thm:spread}
Under the model of \Cref{sec:setup} with symmetric prior
$\pi = 1/2$, informed fraction $\mu$, and flip probability
$\eta \in [0, 1/2]$,
\begin{equation}
\mathrm{spread} \;=\; \ask - \bid \;=\; \mu(1 - 2\eta) \cdot \Delta.
\label{eq:spread}
\end{equation}
\end{theorem}

\begin{proof}
Set $s := \mu(1 - 2\eta)$. The conditional probabilities of true
direction given value are
$P(d = \text{buy} \mid \vH) = (1+\mu)/2$,
$P(d = \text{buy} \mid \vL) = (1-\mu)/2$,
since the informed trader (prob.\ $\mu$) buys deterministically
on $\vH$ and sells on $\vL$, while the noise trader (prob.\
$1-\mu$) buys with probability $1/2$.

For the observed direction:
\begin{align*}
P(\dtil = \text{buy} \mid \vH)
&= P(d = \text{buy} \mid \vH)(1 - \eta)
   + P(d = \text{sell} \mid \vH)\,\eta \\
&= \tfrac{1 + \mu}{2}(1 - \eta) + \tfrac{1 - \mu}{2}\,\eta
 = \tfrac{1 + \mu(1 - 2\eta)}{2}
 = \tfrac{1 + s}{2}.
\end{align*}
By the same calculation,
$P(\dtil = \text{buy} \mid \vL) = (1 - s)/2$.

Under symmetric prior, $P(\dtil = \text{buy}) = 1/2$, so by Bayes
$P(\vH \mid \dtil = \text{buy}) = (1+s)/2$ and
$P(\vH \mid \dtil = \text{sell}) = (1-s)/2$.

Therefore
\begin{align*}
\ask &= \vH \cdot \tfrac{1+s}{2} + \vL \cdot \tfrac{1-s}{2}
      = \tfrac{\vH + \vL}{2} + \tfrac{s \Delta}{2}, \\
\bid &= \vH \cdot \tfrac{1-s}{2} + \vL \cdot \tfrac{1+s}{2}
      = \tfrac{\vH + \vL}{2} - \tfrac{s \Delta}{2},
\end{align*}
yielding $\ask - \bid = s \Delta = \mu(1 - 2\eta) \Delta$. \qed
\end{proof}

\begin{remark}[Sanity: no-privacy limit]
\label{rem:sanity}
Setting $\eta = 0$ recovers the textbook GM spread
$\mu \Delta$. Setting $\eta = 1/2$ gives spread $= 0$:
under perfect privacy noise (50\% flip), the MM's observation
$\dtil$ is independent of $v$, so the posterior collapses to the
prior and the MM posts a single mid-price.
\end{remark}

\begin{remark}[The subsidy is intrinsic to coarse-signal pricing]
\label{rem:competitive}
\begin{sloppypar}
The subsidy is not a feature of the committed-pricing label but of
the maker's coarsened observation. A maker that could condition on
the \emph{true} direction would set
$\ask_{\mathrm{comp}} = \E[v \mid d = \text{buy}]
 = (\vH+\vL)/2 + \mu\Delta/2$, the textbook GM quote with spread
$\mu\Delta$ independent of $\eta$, and would break even. But a
maker that observes only the noisy signal $\dtil$ cannot condition
on $d$: zero profit conditional on its actual information forces
$\ask = \E[v \mid \dtil = \text{buy}]$, exactly the committed
Bayesian quote, with spread $\mu(1-2\eta)\Delta$ and expected loss
$|\pi_M^{\mathrm{GM}}| = \mu\eta\Delta$ against the true flow. The
subsidy is thus the generic cost of pricing on a flipped (coarser)
direction signal, paralleling the Kyle / Gaussian-noise
companion~\cite{nakamura2026privacysubsidy}; it vanishes only if
the maker can price on the un-flipped direction, which the privacy
mechanism precludes.
\end{sloppypar}
\end{remark}

\section{Welfare decomposition: the privacy subsidy}
\label{sec:welfare}

We compute the per-trade expected P\&L for each agent.

\begin{lemma}[Per-agent P\&L formulas]
\label{lem:pl}
Under the equilibrium of \Cref{thm:spread}:
\begin{align}
\pi_I &= +\tfrac{1 - \mu(1 - 2\eta)}{2} \cdot \Delta
      & &\text{(informed)},
      \label{eq:pi-I} \\
\pi_N &= -\tfrac{\mu(1 - 2\eta)}{2} \cdot \Delta
      & &\text{(noise trader)},
      \label{eq:pi-N} \\
\pi_M &= -\mu \eta \Delta
      & &\text{(MM / protocol)}.
      \label{eq:pi-M}
\end{align}
The three components are zero-sum: $\mu \pi_I + (1-\mu) \pi_N +
\pi_M = 0$.
\end{lemma}

\begin{proof}
Set $s = \mu(1 - 2\eta)$ as before. The informed trader, knowing
$v$, trades the profitable direction and realizes
\begin{align*}
\pi_I &= \begin{cases}
\vH - \ask = \Delta/2 - s\Delta/2 & \text{when } v = \vH, \\
\bid - \vL = \Delta/2 - s\Delta/2 & \text{when } v = \vL,
\end{cases}
\end{align*}
so $\pi_I = (1 - s)\Delta/2$, establishing \eqref{eq:pi-I}.

The noise trader, choosing direction uniformly and independently
of $v$, has per-trade gain
\[
\pi_N
= \tfrac{1}{2}(v - \ask) + \tfrac{1}{2}(\bid - v)
= \tfrac{\bid - \ask}{2}
= -\tfrac{\mathrm{spread}}{2}
= -\tfrac{s \Delta}{2},
\]
independent of $v$, establishing \eqref{eq:pi-N}.

By the zero-sum identity (each trade is between trader and MM),
\[
\pi_M = -[\mu\,\pi_I + (1-\mu)\,\pi_N]
      = -\mu \cdot \tfrac{(1-s)\Delta}{2}
        + (1-\mu) \cdot \tfrac{s\Delta}{2}
      = \tfrac{\Delta}{2}(s - \mu).
\]
Substituting $s = \mu(1 - 2\eta)$ gives $s - \mu = -2\mu\eta$, so
$\pi_M = -\mu\eta\Delta$, establishing \eqref{eq:pi-M}. \qed
\end{proof}

\begin{theorem}[Privacy subsidy in GM]
\label{thm:welfare}
Under the model of \Cref{sec:setup}, the MM's expected loss per
trade is
\begin{equation}
|\pi_M^{\mathrm{GM}}| \;=\; \mu \eta \Delta \;\geq\; 0,
\label{eq:subsidy}
\end{equation}
with equality if and only if $\eta = 0$. The quantity
$|\pi_M^{\mathrm{GM}}|$ is the \emph{privacy subsidy in
Glosten-Milgrom}: the per-trade transfer from the protocol/LP
pool to traders, induced by the privacy mechanism.
For protocol break-even, the total per-trade fee revenue must
satisfy $\mathrm{fees}_{\text{trade}} \geq \mu\eta\Delta$.
\end{theorem}

\begin{proof}
\eqref{eq:subsidy} is \eqref{eq:pi-M} of \Cref{lem:pl}; non-
negativity is immediate. Equality at $\eta = 0$ recovers textbook
GM zero-profit. \qed
\end{proof}

\section{Comparative statics}
\label{sec:cs}

\begin{proposition}[Comparative statics]
\label{prop:cs}
The privacy subsidy of \Cref{thm:welfare} satisfies:
\begin{enumerate}[leftmargin=*]
\item \emph{Linear growth}: $|\pi_M^{\mathrm{GM}}|$ is linear in
both $\mu$ and $\eta$, with
$|\pi_M^{\mathrm{GM}}| = \mu\eta\Delta$ for all $(\mu, \eta) \in
(0,1) \times [0, 1/2]$.

\item \emph{Spread vs subsidy trade-off}: the spread $\mu(1-2\eta)
\Delta$ decreases linearly in $\eta$ from $\mu\Delta$ (at
$\eta = 0$) to $0$ (at $\eta = 1/2$), while the subsidy increases
linearly from $0$ to $\mu\Delta/2$.

\item \emph{Informed-trader gain}: $\pi_I = (1-s)\Delta/2$
increases in $\eta$. Privacy noise raises informed-trader profit
by transferring spread reduction.
\end{enumerate}
\end{proposition}

\begin{proof}
Direct calculation from \eqref{eq:spread}, \eqref{eq:pi-I},
\eqref{eq:subsidy}. \qed
\end{proof}

\begin{corollary}[Noise traders also benefit from privacy]
\label{cor:noise-helped}
$\partial \pi_N / \partial \eta = \mu\Delta > 0$ for all
$\eta \in [0, 1/2]$: the uninformed noise traders' expected loss
per trade is strictly decreasing in the privacy parameter.
\end{corollary}

\begin{proof}
Differentiating \eqref{eq:pi-N},
$\partial \pi_N / \partial \eta = \mu\Delta > 0$. \qed
\end{proof}

\begin{sloppypar}
As in the Kyle / Gaussian-noise companion~\cite{nakamura2026privacysubsidy},
\emph{both} trader types gain from privacy and the protocol
bears the entire cost. Privacy reduces informed-trader rent
extraction per unit spread, but because the MM is committed to
the Bayesian rule, the rent loss from a narrower spread is borne
by the MM rather than transferred back to informed traders.
\end{sloppypar}

\begin{remark}[The privacy "gain" is gross-of-fees;
\Cref{cor:noise-helped} is welfare-neutral net-of-fees]
\label{rem:welfare-neutral}
\Cref{cor:noise-helped} is a \emph{gross-of-fees} statement
about the no-fee equilibrium of \Cref{thm:spread}. At the same
equilibrium, the per-trade break-even fee
$f = |\pi_M^{\mathrm{GM}}| = \mu\eta\Delta$ from
\Cref{thm:welfare}, charged flat on every trade, exactly cancels
each trader type's incremental gain over the $\eta = 0$
benchmark. Computing each side directly,
$\pi_I(\eta) - \pi_I(0) = \mu\eta\Delta$ and
$\pi_N(\eta) - \pi_N(0) = \mu\eta\Delta$, so net-of-fees,
\[
  \pi_I - f \;=\; \tfrac{1 - \mu}{2}\Delta
  \quad\text{and}\quad
  \pi_N - f \;=\; -\tfrac{\mu}{2}\Delta,
\]
both equal to their respective classical Glosten--Milgrom values
at $\eta = 0$. The MM is exactly compensated by the per-trade fee
revenue and returns to zero expected profit. Privacy is therefore
\emph{exactly welfare-neutral} under the per-trade break-even fee,
at the partial-equilibrium level (no-fee equilibrium order
arrivals, with fee revenue redistributed to the LP pool). The
full fee-equilibrium analysis, in which a per-trade fee may
distort the informed trader's decision to participate (informed
arrivals cease whenever the per-trade fee exceeds the realized
gain), remains open.
\end{remark}

\begin{remark}[Asymmetric prior left to future work]
\label{rem:asymmetric}
We restrict to symmetric prior $\pi = 1/2$. Asymmetric $\pi$
introduces additional dependence on the joint posterior
$P(v \mid \dtil)$ through
$P(\dtil = \text{buy}) = (1 + s(2\pi - 1))/2 \neq 1/2$, and the
exact closed-form spread and subsidy involve $\pi(1-\pi)$ factors
in the denominator. The qualitative result --- a spread that
shrinks with $\eta$ and a privacy subsidy borne by the protocol
--- is expected to persist, but the precise asymmetric formulae
are left to future work.
\end{remark}

\section{Application to privacy-preserving exchanges}
\label{sec:application}

The flip-noise channel of \Cref{sec:setup} models information loss
in privacy mechanisms that operate on the direction signal of an
arriving trade. We discuss two classes of applications that fit
the model cleanly, then delineate the boundary of applicability
following the discipline established in the Kyle / Gaussian-noise
companion~\cite{nakamura2026privacysubsidy}.

\subsection{Primary application: MPC matching with $\varepsilon$-leakage}
\label{ssec:mpc}

Multi-party-computation (MPC)-based matching engines aim to keep
trade direction private under cryptographic guarantees. When the
MPC protocol exposes a binary direction signal to the matching
engine at a controlled leakage rate --- e.g., via an
$\varepsilon$-differentially-private direction-disclosure
subroutine --- the leakage maps directly to a flip probability.
For the symmetric randomized-response mechanism that saturates the
$\varepsilon$-local-DP bound on the two-class direction
label~\cite{warner1965rr,dwork2014algorithmic}, the flip
probability is $\eta = 1/(1 + e^\varepsilon)$; a general
$\varepsilon$-DP mechanism need not be a symmetric channel, and the
asymmetric case is deferred to future work.

Under this interpretation, the matching engine (acting as the MM)
commits to using the disclosed signal for Bayesian pricing of the
trade. The closed-form spread $\mu(1 - 2\eta)\Delta$ and privacy
subsidy $\mu\eta\Delta$ of
\Cref{thm:spread,thm:welfare} then quantify the cost of the
$\varepsilon$-leakage budget in welfare terms.
\Cref{tab:mpc-fee} tabulates these values for representative
$\eta$.

\begin{table}[H]
\centering
\begin{tabular}{c|c|c}
$\eta$ & $\mathrm{spread}/\Delta = \mu(1-2\eta)$ &
   subsidy $|\pi_M|/\Delta = \mu\eta$ \\
\hline
$0$    & $\mu$ (textbook GM) & $0$ \\
$0.1$  & $0.8 \mu$ & $0.1 \mu$ \\
$0.25$ & $0.5 \mu$ & $0.25 \mu$ \\
$0.4$  & $0.2 \mu$ & $0.4 \mu$ \\
$0.5$  & $0$ (perfect privacy) & $0.5 \mu$ \\
\end{tabular}
\caption{Spread and privacy subsidy under binary-channel privacy
noise, in units of $\Delta$. The protocol's break-even fee floor
is the subsidy column.}
\label{tab:mpc-fee}
\end{table}

\subsection{Secondary: RFQ with side-channel direction inference}
\label{ssec:rfq}

In an idealized request-for-quote (RFQ) design, a dealer prices
without knowing direction; defensive symmetric quotes (a wide
spread) are the standard response. Our flip-noise model does not
apply to this idealized case.

The model does apply, however, to RFQ implementations where a
side channel admits binary direction inference at a controlled
error rate $\eta$. Concrete instances include: an order routing
layer that encrypts direction but leaks a binary signal via
timing or metadata; a dealer-side classifier that infers direction
from trader identity or order patterns at known accuracy
$1 - \eta$; or a committed-mechanism implementation that
deliberately exposes a noised direction label for pricing
fairness. The motivation is more contrived than the MPC case
above, and the boundary between "side channel" and "real
privacy violation" is venue-specific.

This use case is distinct from the Suave-style sealed-bid
order-flow auctions discussed as out-of-scope in the Kyle /
Gaussian-noise companion paper~\cite{nakamura2026privacysubsidy},
which create \emph{temporal} information asymmetry (delayed
reveal) rather than \emph{channel} noise on direction.

\subsection{Mechanisms outside our framework}
\label{ssec:outside}

Three nearby privacy designs do \emph{not} fit our flip-noise
framework and require separate analysis.

\paragraph{Batched aggregation (Penumbra-style).}
Batched-swap designs reveal the exact aggregate of all directions
within a batch window, not a noisy version of any individual
direction. From the matching engine's perspective the aggregate
is a sufficient statistic for the directional information
contained in the batch, and Bayesian pricing on this aggregate
does not generate the kind of flip-channel coarsening we study
here. We conjecture the per-trade subsidy under Bayesian-AMM
pricing on aggregated batches is zero, by analogy with the
Gaussian-flow argument in~\cite{nakamura2026privacysubsidy}; a
discrete-aggregate proof is left to future work.

\paragraph{Sealed-bid with delayed reveal (Suave-style).}
Sealed-bid order-flow auctions create \emph{temporal} information
asymmetry --- the MM observes flow after a delay rather than a
noised version of it. The adverse-selection mechanism is closer
to Loss-Versus-Rebalancing (LVR) than to additive direction
flipping; a continuous-time analysis with explicit time-lag is
required.

\paragraph{Oracle-pegged crossings.}
Designs that match peer orders at an external lit-exchange
midpoint via MPC consume no direction signal at the on-chain
mechanism level. Such oracle-pegged designs fall outside any
GM-style analysis, since pricing is exogenous.

\section{Connection to Loss-Versus-Rebalancing}
\label{sec:lvr}

The privacy subsidy in GM $|\pi_M^{\mathrm{GM}}| = \mu \eta
\Delta$ joins the same conceptual family as
\emph{Loss-Versus-Rebalancing} (LVR), the closed-form
continuous-time AMM adverse-selection cost identified by Milionis
et al.~\cite{milionis2022lvr}, and as the Kyle / Gaussian-noise
privacy subsidy of \cite{nakamura2026privacysubsidy}. All three
are closed-form, per-period welfare quantities measuring an
adverse-selection cost borne by an automated pricing mechanism in
the presence of informed traders. The source of the cost varies
across the three: LVR captures cost from stale prices in
continuous-time AMMs; the Gaussian privacy subsidy captures cost
from Gaussian observation noise; the present discrete privacy
subsidy captures cost from binary-channel direction flipping.

\section{Conclusion}
\label{sec:conclusion}

We derived a closed-form bid-ask spread and welfare decomposition
for the Glosten-Milgrom 1985 sequential-trading model under
binary-symmetric-channel privacy noise on the market maker's
observation of trade direction, with a committed Bayesian-AMM
pricing rule. The equilibrium spread is $\mu(1-2\eta)\Delta$ and
the protocol bears a per-trade welfare loss --- the
\emph{privacy subsidy} --- of $\mu \eta \Delta$. The result
extends the privacy-subsidy concept of
\cite{nakamura2026privacysubsidy} from continuous Gaussian to
discrete two-state microstructure.

\paragraph{Future work.}
Asymmetric flip channels with
$\eta_{\text{buy}\to\text{sell}} \neq \eta_{\text{sell}\to\text{buy}}$
extend the binary symmetric setup. Multi-period sequential
analysis would connect to LVR via repeated GM rounds and may
yield a discrete analog of the additive LVR-plus-privacy
decomposition conjectured in~\cite{nakamura2026privacysubsidy}.
Mechanism design with endogenous $\eta$ treats the exchange
designer's privacy-utility trade-off. Lean 4 mechanization of
the closed-form results across the cluster is a natural
formal-verification target.

\bibliographystyle{splncs04}
\bibliography{references}

@article{glosten1985,
  author  = {Glosten, Lawrence R. and Milgrom, Paul R.},
  title   = {Bid, Ask and Transaction Prices in a Specialist Market with Heterogeneously Informed Traders},
  journal = {Journal of Financial Economics},
  volume  = {14},
  number  = {1},
  pages   = {71--100},
  year    = {1985}
}

@article{glosten1988,
  author  = {Glosten, Lawrence R. and Harris, Lawrence E.},
  title   = {Estimating the components of the bid/ask spread},
  journal = {Journal of Financial Economics},
  volume  = {21},
  number  = {1},
  pages   = {123--142},
  year    = {1988}
}

@article{easleyohara1992,
  author  = {Easley, David and O'Hara, Maureen},
  title   = {Time and the Process of Security Price Adjustment},
  journal = {Journal of Finance},
  volume  = {47},
  number  = {2},
  pages   = {577--605},
  year    = {1992}
}

@inproceedings{brahma-bayesian-mm,
  author    = {Brahma, Aseem and Chakraborty, Mithun and Das, Sanmay and Lavoie, Allen and Magdon-Ismail, Malik},
  title     = {A {B}ayesian Market Maker},
  booktitle = {Proceedings of the 13th ACM Conference on Electronic Commerce (EC 2012)},
  year      = {2012}
}

@article{das2005learning,
  author  = {Das, Sanmay},
  title   = {A Learning Market-Maker in the {G}losten--{M}ilgrom Model},
  journal = {Quantitative Finance},
  volume  = {5},
  number  = {2},
  pages   = {169--180},
  year    = {2005}
}

@article{huddart2001disclosure,
  author  = {Huddart, Steven and Hughes, John S. and Levine, Carolyn B.},
  title   = {Public Disclosure and Dissimulation of Insider Trades},
  journal = {Econometrica},
  volume  = {69},
  number  = {3},
  pages   = {665--681},
  year    = {2001}
}

@article{touzo2020information,
  author  = {Touzo, L{\'e}o and Marsili, Matteo and Zagier, Don},
  title   = {Information thermodynamics of financial markets: the {G}losten--{M}ilgrom model},
  journal = {Journal of Statistical Mechanics: Theory and Experiment},
  year    = {2021},
  note    = {arXiv:2010.01905}
}

@article{carmier2022thermo,
  author  = {Carmier, Pierre},
  title   = {Generalized second law of thermodynamics in the {G}losten--{M}ilgrom model},
  journal = {arXiv preprint arXiv:2209.15429},
  year    = {2022}
}

@article{zhu2014darkpools,
  author  = {Zhu, Haoxiang},
  title   = {Do Dark Pools Harm Price Discovery?},
  journal = {Review of Financial Studies},
  volume  = {27},
  number  = {3},
  pages   = {747--789},
  year    = {2014}
}

@article{buti2017darkpool,
  author  = {Buti, Sabrina and Rindi, Barbara and Werner, Ingrid M.},
  title   = {Dark Pool Trading Strategies, Market Quality and Welfare},
  journal = {Journal of Financial Economics},
  volume  = {124},
  number  = {2},
  pages   = {244--265},
  year    = {2017}
}

@article{milionis2022lvr,
  author  = {Milionis, Jason and Moallemi, Ciamac C. and Roughgarden, Tim and Zhang, Anthony Lee},
  title   = {Automated Market Making and Loss-Versus-Rebalancing},
  journal = {arXiv preprint arXiv:2208.06046},
  year    = {2022}
}

@article{warner1965rr,
  author  = {Warner, Stanley L.},
  title   = {Randomized Response: A Survey Technique for Eliminating Evasive Answer Bias},
  journal = {Journal of the American Statistical Association},
  volume  = {60},
  number  = {309},
  pages   = {63--69},
  year    = {1965}
}

@article{dwork2014algorithmic,
  author  = {Dwork, Cynthia and Roth, Aaron},
  title   = {The Algorithmic Foundations of Differential Privacy},
  journal = {Foundations and Trends in Theoretical Computer Science},
  volume  = {9},
  number  = {3--4},
  pages   = {211--407},
  year    = {2014}
}

@misc{nakamura2026privacysubsidy,
  author        = {Nakamura, Yuki},
  title         = {The Privacy Subsidy: Kyle's $\lambda$ under Noise-Perturbed Order-Flow Observation},
  year          = {2026},
  eprint        = {2605.15746},
  archivePrefix = {arXiv},
  primaryClass  = {cs.GT}
}

\end{document}